\documentclass[conference, 10pt, twocolumn]{ieeeconf}

\IEEEoverridecommandlockouts
\usepackage[usenames]{color}
\usepackage{enumerate}
\usepackage{url}
\usepackage{subfigure}
\usepackage{amsfonts,mathrsfs}
\usepackage{amssymb,amsmath}
\usepackage{verbatim}
\usepackage{acronym}
\usepackage{graphicx}

\def\fskip#1{}

\newtheorem{theorem}{Theorem}

\newtheorem{corollary}{Corollary}

\newtheorem{definition}{Definition}
\newtheorem{example}{Example}

\newtheorem{lemma}{Lemma}

\newtheorem{proposition}[theorem]{Proposition}
\newtheorem{remark}{Remark}

\renewcommand{\theenumi}{\arabic{enumi}}

\def\1{{\bf 1}}

\def\N{\mathbb{N}}

\def\N{\mathcal{N}}

\def\R{\mathbb{R}}

\newcommand{\remove}[1]{}

\def\argmin{\mathop{\rm argmin}}

\begin{document}
\title{Game-Theoretic Analysis of the Hegselmann-Krause Model for Opinion Dynamics in Finite Dimensions*}
\author{\authorblockN{Seyed Rasoul Etesami, Tamer Ba\c{s}ar}
  \authorblockA{Coordinated Science Laboratory, University of Illinois at Urbana-Champaign,  Urbana, IL 61801\\
     Email: (etesami1, basar1)@illinois.edu}
\thanks{*Research supported in part by the ``Cognitive \& Algorithmic Decision Making" project grant through the College of Engineering of the University of Illinois, and in part 
by AFOSR MURI Grant FA 9550-10-1-0573 and NSF grant CCF 11-11342.}
}
\maketitle
\begin{abstract}
We consider the Hegselmann-Krause model for opinion dynamics and study the evolution of the system under various settings. We first analyze the termination time of the synchronous Hegselmann-Krause dynamics in arbitrary finite dimensions and show that the termination time in general only depends on the number of agents involved in the dynamics. To the best of our knowledge, that is the sharpest bound for the termination time of such dynamics that removes dependency of the termination time from the dimension of the ambient space. This answers an open question in \cite{bhattacharyya2013convergence} on how to obtain a tighter upper bound for the termination time. Furthermore, we study the asynchronous Hegselmann-Krause model from a novel game-theoretic approach and show that the evolution of an asynchronous Hegselmann-Krause model is equivalent to a sequence of best response updates in a well-designed potential game. We then provide a polynomial upper bound for the expected time and expected number of switching topologies until the dynamic reaches an arbitrarily small neighborhood of its equilibrium points, provided that the agents update uniformly at random. This is a step toward analysis of heterogeneous Hegselmann-Krause dynamics. Finally, we consider the heterogeneous Hegselmann-Krause dynamics and provide a necessary condition for the finite termination time of such dynamics. In particular, we sketch some future directions toward more detailed analysis of the heterogeneous Hegselmann-Krause model.
\end{abstract}

\smallskip

\section{Introduction}
Opinion formation in social networks is an important area of research that has attracted a lot of attention in recent years in a wide range of disciplines, such as psychology, economics, political science, and electrical and computer engineering. A natural question that commonly arises in all those areas is the extent to which one can predict the outcome of the opinion formation of entities under some complex interaction process running among these social actors. Consensus problems in which a set of agents are trying to achieve the same goal have been addressed by many researchers, such as \cite{consensusb,jadbabaie2003coordination,olfati2004consensus,Kashyab,finitecp,degroot1974reaching,alex,gossipca,optimalc,finitec,etesami2013quantized}. In such problems, which are still an active area of research, the goal is to achieve a certain agreement among agents. However, there are many of situations in which there is neither a desire for consensus nor any tendency for the underlying process to approach a common outcome. In fact, such situations frequently emerge in the context of political elections and product marketings when there are multiple candidates or product choices to be selected among. Those facts have motivated researchers to study disagreement along with consensus. 

\smallskip
One of the first studies that considers disagreement beside consensus was undertaken by Friedkin and Johnsen \cite{friedkin1999social}, whose model was later extended by Hegselmann and Krause in \cite{hegselmannk}, in the sense that \cite{hegselmannk} relaxes the assumption of time-invariant influence weights among the agents. More precisely, the Hegselmann-Krause dynamics allow the influence weights to be a function of not only time, but also the states. It is worth noting that although such extensions make the analysis of Hegselmann-Krause dynamics mathematically much more complicated but interesting, one may argue that the assumption of influence weights depending on the evolving opinion distance (which is the case in the Hegselmann-Krause dynamics) is questionable from a practical point of view, given the literature in experimental social psychology, e.g., see  \cite{stiff2003persuasive,friedkin2011social}, where social psychologists have long been intrigued by the hypothesis that opinion differences reliably predict direct relations of interpersonal influence. Still, a rigorous analysis of the Hegselmann-Krasue dynamics is both theoretically and practically important. The theoretical aspects is that it allows us to develop novel tools useful to study more complex time and state dependent evolutionary dynamics and elaborate on their connections with other fields. The practical aspect is that, other than applications in the modeling of opinion dynamics, the model has applications in the robotics rendezvous problem in plane and space \cite{Bullo}. Accordingly, we consider the Hegselmann-Krause model in $\R^{d}$, where $d\ge 1$.   

\smallskip
In the Hegselmann-Krause model, a finite number of agents frequently update their opinions based on the possible interactions among them. The opinion of each agent in this model is captured by a scalar quantity in one dimension or a vector in Euclidean space $\mathbb{R}^{d>1}$ in higher dimensions. In fact, because of the conservative nature of social entities, each agent in this model communicates only with those whose opinions are closer to him and lie within a certain level of his confidence (bound of confidence), where the distance between agents' opinions is measured by the Euclidian norm in the ambient space. Depending on whether the bound of confidence is the same for all the agents or not, one can distinguish two different types of dynamics, known as \textit{homogeneous} and \textit{heterogeneous}, respectively. Moreover, the updating process of the agents may be synchronous, meaning that all the agents update simultaneously, or asynchronous, where the agents update in turn. Although at first glance the differences among these four types of dynamics may seem negligible, in fact, their outcomes are substantially different, such that most of the results from one cannot be carried over to the others \cite{lorenzt,lorenz2010heterogeneous,mirtabatabaei2012opinion}. In particular, because of the extra freedom for the agents' movements in higher dimensions, analyzing such dynamics for dimensions higher than one is considerably more complex than for one dimension \cite{etesami2013termination,bhattacharyya2013convergence,chazelle2011total}.  

It is known that synchronous homogeneous Hegselmann-Krause dynamics will terminate after finitely many steps 
\cite{hegselmannk,lorenzt}. The same model 
has also been used for distributed rendezvous in a robotic network \cite{Bullo,robot-krause}. In the model, depending on the initial profile and the confidence bound, the final state may or may not 
be a consensus. The existing studies on the behavior of the Hegselmann-Krause model in one dimension where the agents' opinions are scalars can be found in \cite{lorenzc}. It was shown in \cite{Bullo} that the termination time of the Hegselmann-Krause dynamics in one dimension is at least $O(n)$, where $n$ is the number of agents, and at most $O(n^3)$ \cite{bhattacharyya2013convergence,mohajer2012convergence}. Moreover, the stability and the termination time of such dynamics in higher dimensions were studied in \cite{nedic2012multi,etesami2013termination}, and the work in \cite{etesami2013termination} bounds the termination time of such dynamics using the number of isolated agents through the evolution of the dynamics. In a recent work of Bhattacharyya et al. \cite{bhattacharyya2013convergence}, a polynomial upper bound of $O(n^{10}d^2)$ was given for such dynamics in higher dimensions, but leaving the dependency of such a bound on the dimension of ambient space as an open problem. In this work, we improve the upper bound to $O(n^8)$ and show that the termination time is, indeed, independent of the dimension of the ambient space.   

The asynchronous homogeneous Hegselmann-Krause model was considered in \cite{touri2014endogenous}, where the authors were able to establish stability of this model using a proper quadratic comparison function when the probability of updating for each agent is uniformly bounded from below by some positive constant $p>0$. In this paper, we model the evolution of such dynamics as a sequence of best response updates in a potential game and provide a polynomial upper bound for the maximum expected switching topologies and the expected time it takes for the dynamics to reach an arbitrarily small neighborhood of its steady state provided that the agents update uniformly at random. We refer readers to \cite{marden2009cooperative} and \cite{rantzer2008using} for some of the possible connections between control of distributed systems and potential games. Furthermore, the synchronous heterogeneous Hegselmann-Krause model was studied in \cite{lorenz2010heterogeneous}, and recently in \cite{mirtabatabaei2012opinion}, where the authors conjecture that the number of switching topologies throughout the dynamics must be finite. In fact, our analysis for an asynchronous homogeneous Hegselmann-Krause model here is a step toward more detailed analysis of the heterogeneous model using an appropriate potential function over directed graphs \cite{zhang2012lyapunov,olfati2004consensus}. Furthermore, numerous simulation results have been conducted to study and explore the evolutionary properties of the Hegselmann-Krause dynamics under various settings. For more information, we refer the reader to \cite{lorenzt,lorenz2010heterogeneous,hegselmannk,julient}.

This paper is organized as follows. In Section~\ref{sec:HKdyn}, we review the Hegselmann-Krause dynamics under various settings. In Section~\ref{sec:prelimresults}, 
we develop some preliminary results and mention some existing results for later use. In Section \ref{sec:mainresults-one} we consider the synchronous Hegselmann-Krause model in arbitrary finite dimensions and provide a polynomial upper bound for the termination time, independent of the dimension of the opinion space. That not only improves on the previous bounds, but also removes the dependency of the termination time on the dimension of the ambient space. In Section~\ref{sec:mainresults-two}, we model the asynchronous Hegselmann-Krause model as a potential game and provide its corresponding potential function. Using that function, we bound the expected number of switching topologies of the network when the agents update their opinions uniformly at random. Moreover, we provide an upper bound for the expected number of steps until the agents reach a $\delta$-neighborhood of their steady state for some $\delta>0$. We also directly show strategic equivalence of the game to a team problem. In Section \ref{sec:mainresults-three} we turn our attention to the heterogeneous Hegselmann-Krause model and provide a necessary condition for such dynamics to terminate in finite time. In Section \ref{sec:discussion}, using the tools developed in this work, we discuss some of the possible future directions toward more detailed analysis of heterogeneous Hegselmann-Krause dynamics. We conclude the paper with the final remarks of Section~\ref{sec:conclusion}. 

\textbf{Notations}: 
For a positive integer $n$, we let $[n]:=\{1,2,\ldots,n\}$. For a vector $v\in \R^n$, we let $v_i$ be the $i$th entry of $v$. We say that $v$ is \textit{stochastic} if $v_i\geq 0$ for all $i\in [n]$ and $\sum_{i=1}^nv_i=1$. Similarly, for a matrix $A$, we let $A_{ij}$ be the $ij$th entry of $A$. We say that $A$ is \textit{stochastic} (or \textit{row-stochastic}) if each of its rows is stochastic, and we let $\min^+A=\min_{i,j}\{A_{ij}| A_{ij}>0\}$. We use $A_i$ to denote the $i$th row of $A$. We use $A'$ to denote the transpose of a matrix $A$, and $\|v\|$ to denote the Euclidean norm of a vector $v$. We let the consensus vector $\bold{1}$ be a vector of unit size ($\|\bold{1}\|=1$) with equal entries. For a matrix $A$ with real eigenvalues, we let $\lambda_2(A)$ be its second smallest eigenvalue. A \textit{scrambling matrix} is a stochastic matrix such that the inner product of each pair of its rows is positive. For a vector $y$ we use $conv(y)$ to show the convex hull of its components and $diam(conv(y))=\max_{p,q\in conv(y)}\|p-q\|$. We define the distance between two sets $P, Q\subseteq \R^n$ to be $dist(P,Q)=\inf_{p\in P, q\in Q }\|p-q\|$. For a graph $\mathcal{G}$, we let $\mathcal{A}_{\mathcal{G}}$ be its adjacency matrix and $\mathcal{D}_\mathcal{G}$ be a diagonal matrix whose diagonal entries are equal to the degree of the nodes in the graph. Moreover, we use $\mathcal{L}_{\mathcal{G}}=A_{\mathcal{G}}-\mathcal{D}_\mathcal{G}$ to denote the Laplacian of that graph. Finally, we use $|S|$ to denote the cardinality of a finite set $S$.

\bigskip
\section{Hegselmann-Krause Dynamics}\label{sec:HKdyn}
In this section we describe the discrete-time Hegselmann-Krause opinion dynamics model as introduced in \cite{hegselmannk}.

Let us assume that we have a set of $n$ agents $[n]=\{1,\ldots,n\}$ and we want to model the interactions among their opinions. It is assumed that at each time $t=0, 1, 2, \ldots$, the opinion of agent $i\in[n]$ can be represented by a vector $x_{i}(t)\in \R^d$ for some $d\geq 1$. According to that model, the evolution of opinion vectors can be modeled by the following discrete-time dynamics: 
\begin{align}\label{eqn:DynamicForm}
x(t+1)=A(t,x(t),\vec{\epsilon})x(t),
\end{align}
where $A(t,x(t),\vec{\epsilon})$ is an $n\times n$ row-stochastic matrix and $x(t)$ is the $n\times d$ matrix such that its $i$th row contains the opinion of the $i$th agent at time $t=0, 1, 2, \ldots$, 
i.e.,\ it is equal to $x_{i}(t)$. We refer to $x(t)$ as the \textit{opinion profile} at time $t$. The entries of $A(t,x(t),\vec{\epsilon})$ are functions of time step $t$, current profile $x(t)$, confidence vector $\vec{\epsilon}=(\epsilon_1,\epsilon_2,\ldots,\epsilon_n)>0$ and an updating scheme. The parameters $\epsilon_i, i\in [n]$ are referred to as the 
\textit{confidence bounds}. In the homogeneous case of the dynamics, we assume that $\epsilon_i=\epsilon,\  \forall i\in [n]$ for some $\epsilon>0$, while in the heterogeneous model, different agents may have different bounds of confidence. Our focus in this paper is mainly on the homogeneous model, but we also analyze the heterogeneous case toward the end, in Section \ref{sec:mainresults-three}. For the sake of simplicity of notation and for a fixed $x(0)\in \mathbb{R}^{n\times d}$, we drop the dependency of $A(t,x(t),\vec{\epsilon})$ on $x(t)$ and $\epsilon$ and simply write $A(t)$. In what follows next, we distinguish two different versions of Hegselmann-Krause dynamics.

\subsection{Synchronous Hegselmann-Krause Model}
In the synchronous Hegselmann-Krause model, each agent $i$ updates its value at time $t=0, 1, 2, \ldots$, by averaging 
its own value and the values of all the other agents that are in its $\epsilon$-neighborhood at time $t$. To be more specific, given a profile $x(t)$ 
at time $t$, define the matrix $A(t)$ in \eqref{eqn:DynamicForm} by:
\begin{align}\label{eq:synchronous-Hegselmann-Krause}
A_{ij}(t) &= \begin{cases} \frac{1}{|\N_i(t)|} & \mbox{if } j\in \N_i(t), \\
 0 & \mbox{ else}, \end{cases}
\end{align}
where 
$\N_i(t)$ is the set of agents in the $\epsilon$-neighborhood of agent $i$, i.e.,\
 \[\N_i(t)=\{j\in[n]\mid \|x_i(t)-x_j(t)\|\leq \epsilon\}.\]

\smallskip
\subsection{Asynchronous Hegselmann-Krause Model}
In the asynchronous case and at each time instant $t=0,1,2,\ldots$, only one agent, namely $i^*$, updates its value to the average of its neighbors, while the others remain unchanged. Selection of such an agent may be at random or based on some predefined order. In this paper, we assume that the agents are chosen uniformly at random to update their opinions. In that case the updating matrix $A(t,x(t),\vec{\epsilon})$ given in \eqref{eqn:DynamicForm} can be written as 

\begin{align}\label{eq:asynchronous-Hegselmann-Krause}
A_{ij}(t) = \begin{cases} \frac{1}{|\N_{i^*}(t)|} & \mbox{if} \ \ i=i^*, j\in \N_{i^*}(t), \\ 
1 & \mbox{if} \ \ i=j\neq i^* \\
0 & \mbox{else}, \end{cases}
\end{align} 
where here we have assumed that agent $i^*$ updates its opinion at time $t$.

\bigskip
\begin{remark}
In the heterogeneous Hegselmann-Krause model, each agent $i$ is able to observe only its $\epsilon_i$-neighborhood, and we have 
\[\N_i(t,\epsilon_i)=\{j\in[n]\mid \|x_i(t)-x_j(t)\|\leq \epsilon_i\}.\]
\end{remark}

\smallskip
\begin{remark}
There are other types of Hegselmann-Krause dynamics where the evolution of dynamics is subject to noise or perturbation in the system or when the agents are truth seekers in the sense that they are attracted by the truth by a positive amount \cite{pineda2013noisy,kurz2011hegselmann}. Moreover, the continuous version of the Hegselmann-Krause model, in which a continuum of opinions are involved in the dynamics, has been considered in \cite{hendrickx2013symmetric,blondel20072r,lorenz2007continuous}.   
\end{remark}

\smallskip
\begin{remark}
As can be seen from the above formulations, the Hegselmann-Krause dynamics do not preserve the opinion average of the agents, and the evolution of the system strongly depends on the history and the states, which may switch between different topologies. In fact, it is not possible to determine the topology of the network at the current time, unless one can observe the state of the system in the previous time step. Those facts make the analysis of such dynamics much more complicated than analysis of the average-preserving dynamics with fixed topology.   
\end{remark}

\bigskip
\section{Preliminary Results}\label{sec:prelimresults}
In this section, we briefly discuss some preliminary results and provide some definitions that will be used to prove our main results.

\begin{lemma}\label{lemma:perron-frobenius}
(\textit{Perron-Frobenius for Laplacians \cite{saloff1997lectures}}): Let $\mathcal{L}$ be a matrix with non-positive off-diagonal entries such that the graph of the non-zero off-diagonal entries is connected. Then the smallest eigenvalue has multiplicity 1, and the corresponding eigenvector is strictly positive.
\end{lemma}

\smallskip
Next, we state Cheeger's inequality, which relates the spectral gap of the Laplacian matrix to the expansion of its corresponding graph.

\smallskip
\begin{lemma}[Cheeger's Inequality \cite{godsil2001algebraic}]\label{lemma:cheegers-inequality}
Let $\mathcal{G}=(\mathcal{V},\mathcal{E})$ be an undirected graph with Laplacian matrix $\mathcal{L}$. Moreover, for a subset of vertices $S\subseteq\mathcal{V}$ let $e(S,S^c)$ denote the number of edges with one vertex in $S$ and one vertex in its complement $S^c$. Defining the cut ratio $\Phi(S)=\frac{e(S,S^c)}{|S||S^c|}$ and isoperimetric number of $\mathcal{G}$ by $\Phi=\min_{S\subset \mathcal{V}}\Phi(S)$, we have
\begin{align}\nonumber 
\frac{\Phi^2}{2d_{max}}\leq \lambda_2(\mathcal{L})\leq 2\Phi,
\end{align}
where $d_{max}$ denotes the maximum degree of the graph $\mathcal{G}$ and $\lambda_2(\mathcal{L})$ is the second smallest eigenvalue of the Laplacian $\mathcal{L}$. 
\end{lemma}

\bigskip
\begin{lemma}(\textit{Courant-Fischer Formula \cite{horn2012matrix}})\label{lemma:Courant-Fischer-Formula}
Let $A$ be an $n\times n$ symmetric matrix with eigenvalues $\lambda_1\leq \lambda_2 \leq \ldots,\leq \lambda_n$ and corresponding eigenvectors $v_1,\ldots,v_n$. Moreover, for $1\leq k \leq n$, let $S_k$ denote the span of $v_1,\ldots,v_k$ (with $S_0=\{0\}$), and let $S_k^{^\bot}$ denote the 
orthogonal complement of $S_k$, i.e., $S_k^{^\bot}=\{v\in \mathbb{R}^n| v'u=0, \forall u\in S_k\}$. Then 
\begin{align}\nonumber
\lambda_k=\min_{\substack{\|x\|=1 \\ x\in S_{k-1}^{^\bot}}}x'Ax.
\end{align}
\end{lemma}

\smallskip
\begin{corollary}[Rayleigh-Quotient \cite{horn2012matrix}]\label{cor:Rayleigh Quotient}
Let $\mathcal{G}=(\mathcal{V},\mathcal{E})$ be a graph and $\mathcal{L}$ be the Laplacian of $\mathcal{G}$. We already know from the Perron-Frobenius lemma (Lemma \ref{lemma:perron-frobenius}) that the smallest eigenvalue is $\lambda_1=0$ with eigenvector $v_1=1$. By the Courant-Fischer Formula, we get
\begin{align}\nonumber
\lambda_2(\mathcal{L})=\min_{\substack{\|x\|=1 \\ x\bot\bold{1}}}x'\mathcal{L}x.
\end{align}
\end{corollary}

\bigskip
\begin{lemma}\label{lemm:convex-hull}
Suppose $C$ is a stochastic matrix and $y=Cx$; then
\begin{align}\nonumber
diam(conv(y))\leq (1-\mu(C)) diam(conv(x))
\end{align}
where $\mu(C)=\min_{i\neq j}(\sum_{k=1}^{n}\min (c_{ik},c_{jk}))$.
In particular, when $C$ is a scrambling matrix with  \ $\min^+C\geq \delta$, then we can say $\mu(C)\ge \delta$, or $diam(conv(y))\leq (1-\delta) diam(conv(x)).$
\end{lemma}
\begin{proof}
A short proof of the above lemma can be found in \cite{geometric}.
\end{proof}

In fact, one of the fundamental concepts and properties of the synchronous Hegselmann-Krause dynamics that will be used extensively throughout this paper is that the dynamics admit a quadratic Lyapunov function \cite{blondel20072r,TouriNedich:CDC2012,roozbehani2008lyapunov}.

\bigskip
\section{Synchronous Multidimensional Hegselmann-Krause Dynamics}\label{sec:mainresults-one}

In this section we consider the homogeneous synchronous Hegselmann-Krause model as was introduced in \eqref{eq:synchronous-Hegselmann-Krause}. We start our discussion by introducing some notation that will be used throughout this section.

\smallskip
\begin{definition}
We say that a time instance $t$ is a \textit{merging time} for the dynamics if two agents with different opinions move to the same place.
\end{definition}

\smallskip
Based on that definition, we can see that if two agents $i$ and $j$ merge at time instant $t$, then they will have the same opinion at time $t+1$ and onward, while their common opinion may vary with time. Moreover,  prior to the termination time of the dynamics, we cannot have more than $n$ merging times, since there are $n$ agents in the model. In what follows next, we define the notions of termination time and communication graphs.
\smallskip
\begin{definition}\label{def:T_n}
For every set of $n\ge 1$ agents we define the \textit{termination time} $T_n$ of the synchronous Hegselmann-Krause dynamics to be the maximum number of iterations before steady state is reached over all the initial profiles.
\end{definition}

\smallskip
\begin{definition}\label{def:trivial-component-communication-graph}
Given an opinion profile at time $t$, we associate with that opinion profile an undirected graph $\mathcal{G}(t)=([n],\mathcal{E}(t))$ where the edge $(i,j)\in \mathcal{E}(t)$ if and only if $i\in \N_j(t)$. We refer to such a graph as the \textit{communication graph} or \textit{communication topology} of the dynamics at time step $t$. Furthermore, a connected component of the communication graph is called $\delta$-\textit{trivial} for some $\delta>0$, if all the agents in that component lie within a distance of at most $\delta$ from each other.  
\end{definition}

\smallskip
\begin{remark}\label{rem:trivial-component}
From Definition \ref{def:trivial-component-communication-graph}, it is not hard to see that for any $\delta<\epsilon$, a $\delta$-trivial component forms a complete component (clique) in the communication topology of the dynamics. In particular, if there is such a $\delta$-trivial component at some time $t$, then in the next time step, all the agents in that component will merge to the same opinion.  
\end{remark}

\smallskip 
In our earlier work \cite{etesami2013termination}, we were able to analyze the termination time of the Hegselmann-Krause dynamics based on the number of isolated agents throughout the dynamics. 

\bigskip
\begin{theorem}\label{thm:termination-time-singletons}
For the termination time $T_n$ of the synchronous Hegselmann-Krause dynamics in $\mathbb{R}^d$, we have:
\begin{align}\nonumber
\sum_{t=0}^{T_n}(\frac{1}{2})^{|S_0(t)|}<8n^6,
\end{align}
where $S_0(t)$ denotes the set of agents (singletons) at time $t$ who do not observe any opinions other than them inside their neighborhood, i.e., $i\in S_0(t)$ if and only if $\N_i(t)=\{x_i(t)\}$. 
\end{theorem}
\begin{proof}
A proof can be found in \cite{etesami2013termination}.
\end{proof}

\smallskip
As a particular result of Theorem \ref{thm:termination-time-singletons}, if for a particular instance of the dynamics, the agents maintain the connectivity throughout the dynamics, we conclude that $T_n=O(n^6)$. In fact, Theorem \ref{thm:termination-time-singletons} gives us the idea that the termination time of the dynamics greatly depends on the connectivity of the underlying graph and should be independent of the dimension of the opinions ($d$). In this paper, we resolve that problem and show that indeed, the termination time is independent of the dimension. That answers one of the open questions raised in $\cite{bhattacharyya2013convergence}$ related to the existence of a tighter polynomial upper bound independent of the dimension of the opinion space. For that purpose, we utilize a quadratic Lyapunov function that was introduced earlier in \cite{roozbehani2008lyapunov}. 
\smallskip
\begin{lemma}\label{lemm:Lyapunov-mardavij}
Let $V(t)\!=\!\sum_{i,j\in [n]}\min\{\|x_i(t)\!-\!x_j(t)\|^2,\epsilon^2\}$. Then $V$ is non-increasing along the trajectory of the synchronous Hegselmann-Krause dynamics. In particular, we have
\begin{align}\nonumber
V(t)-V(t+1)\ge 4\sum_{\ell=1}^{n}\|x_{\ell}(t+1)-x_{\ell}(t)\|^2.
\end{align}
\end{lemma} 
\begin{proof}
A proof can be found in \cite{roozbehani2008lyapunov}.
\end{proof}  
\smallskip

In the following theorem, we provide a lower bound for the amount of decrease of the above Lyapunov function as long as there exists one non-$\epsilon$-trivial component in the dynamics. 

\bigskip 
\begin{theorem}
The termination time of the synchronous Hegselmann-Krause dynamics in arbitrary finite dimensions is independent of the dimension and is bounded from above by $T_n\leq n^8+n$.
\end{theorem}
\begin{proof}
Let us assume that the opinion profile $x(t)=(x_1(t),x_2(t), \ldots, x_n(t))'$ is not an equilibrium point of the dynamics and that time $t$ is not a merging time. Therefore, without loss of generality, we may assume that the communication graph at time $t$ is connected with a non-$\epsilon$-trivial component; otherwise, we can restrict ourselves to one of the non-$\epsilon$-trivial components. (Note that such a non-$\epsilon$-trivial component exists, because of Remark \ref{rem:trivial-component} and the fact that $t$ is not a merging time.) By projecting each individual column of $x(t)$ to the consensus vector $\bold{1}$ we can write
\begin{align}\label{eq:consensus-projection}
x(t)=\Big[c_1\bold{1}|c_2\bold{1}|\ldots|c_d\bold{1}\Big]+\Big[\bar{c}_1\bold{1}^{\!\!\!\!^{\bot (1)}}|\bar{c}_2\bold{1}^{\!\!\!\!^{\bot (2)}}|\ldots|\bar{c}_d\bold{1}^{\!\!\!\!^{\bot (d)}}\Big],
\end{align}
where $\bold{1}^{\!\!\!\!^{\bot(k)}}, k=1,\ldots,d$ are column vectors of unit size that are orthogonal to the consensus vector, i.e., $\bold{1}'\bold{1}^{\!\!\!\!^{\bot (k)}}=0$, and $c_k,\bar{c}_k, k=1,\ldots,d$ are coefficients of projection of the $k$th column of $x(t)$ on $\bold{1}$ and $\bold{1}^{\!\!\!\!^{\bot(k)}}$, respectively. 

Now we claim that $\sum_{k=1}^{d}\bar{c}_k^2>\frac{\epsilon^2}{4}$. Otherwise, we show that every two agents $x_i(t)$ and $x_j(t)$ must lie within a distance of at most $\epsilon$ from each other, which is in contrast with the assumption that the component is a non-$\epsilon$-trivial component. In fact, if $\sum_{k=1}^{d}\bar{c}_k^2\leq\frac{\epsilon^2}{4}$, we can write,
\begin{align}\label{eq:c-bar-epsilon-trivial}
\|x_i(t)-x_j(t)\|^2&=\sum_{k=1}^{d}\bar{c}^2_k\big(\bold{1}_i^{\!\!\!\!^{\bot(k)}}-\bold{1}_j^{\!\!\!\!^{\bot(k)}}\big)^2\cr 
&\leq  2\sum_{k=1}^{d}\bar{c}^2_k\big((\bold{1}_i^{\!\!\!\!^{\bot(k)}})^2+(\bold{1}_j^{\!\!\!\!^{\bot(k)}})^2\big) \cr &\leq 2\sum_{k=1}^{d}\bar{c}^2_k\big(\|\bold{1}^{\!\!\!\!^{\bot(k)}}\|^2+\|\bold{1}^{\!\!\!\!^{\bot(k)}}\|^2\big)\cr 
&=4\sum_{k=1}^{d}\bar{c}^2_k\leq\epsilon^2, 
\end{align} 
where the first equality is due to the decomposition given in \eqref{eq:consensus-projection} and the second equality is valid since the vectors $\bold{1}^{\!\!\!\!^{\bot(k)}}, k=1\ldots,d$, are of unit size. The contradiction shows that $\sum_{k=1}^{d}\bar{c}_k^2>\frac{\epsilon^2}{4}$. 

Next, we notice that $x(t+1)=A(t)x(t)$, where $A(t)$ is the stochastic matrix defined in \eqref{eq:synchronous-Hegselmann-Krause}. Using \eqref{eq:consensus-projection} we can write,
\begin{align}\label{eq:x(t)-x(t+1)-decomposition}
x(t)&-x(t+1)=(I-A(t))x(t)\cr 
&\qquad =\Big[\bar{c}_1(I-A(t))\bold{1}^{\!\!\!\!^{\bot (1)}}|\ldots|\bar{c}_d(I-A(t))\bold{1}^{\!\!\!\!^{\bot (d)}}\Big],
\end{align}
where the equality holds since $\bold{1}$ belongs to the null space of $I-A(t)$. In particular, we have,
\begin{align}\label{eq:decrease-sum-quadratics}
\sum_{\ell=1}^{n}\|x_{\ell}(t)-x_{\ell}(t\!+\!1)\|^2&=\sum_{\ell=1}^{n}\sum_{k=1}^{d}\big(x_{\ell k}(t)-x_{\ell k}(t\!+\!1)\big)^2\cr 
&=\sum_{k=1}^{d}\Big(\sum_{\ell=1}^{n}\big(x_{\ell k}(t)-x_{\ell k}(t\!+\!1)\big)^2\Big)\cr &=\sum_{k=1}^{d}\bar{c}^2_k\|(I-A(t))\bold{1}^{\!\!\!\!^{\bot (k)}}\|^2,
\end{align}
where in the last equality we have used \eqref{eq:x(t)-x(t+1)-decomposition}.
Let us assume that $Q(t)=(I-A(t))'(I-A(t))$. It is not hard to see that $Q(t)$ is a positive semidefinite matrix. Moreover, 0 is an eigenvalue of $Q$ with multiplicity one, corresponding to the eigenvector $\bold{1}$. To see that, let us assume that there exists another vector $v$, such that $Q(t)v=0$. Multiplying that equality from the left by $v'$, we get $\|(I-A(t))v\|^2=0$, and hence $(I-A(t))v=0$. Since by the Perron-Frobenius lemma (Lemma \ref{lemma:perron-frobenius}), $\bold{1}$ is the only unit eigenvector of $I-A(t)$ corresponding to eigenvalue 0, we conclude that $v=\alpha\bold{1}$ for some $\alpha\in \mathbb{R}$. In other words, $\bold{1}$ is the only unit eigenvector of $Q(t)$ corresponding to eigenvalue 0. Moreover, $Q(t)$ is a symmetric real-valued matrix and, hence, diagonalizable, where $\bold{1}$ is its only eigenvector corresponding to eigenvalue 0. That shows that the multiplicity of the eigenvalue 0 in $Q(t)$ is exactly one. 

\smallskip
Let us use $\lambda_2(Q(t))$ to denote the second smallest eigenvalue of $Q(t)$. By the above argument, it must be strictly positive. Using the Courant-Fischer lemma (Lemma \ref{lemma:Courant-Fischer-Formula}), we get $\lambda_2(Q(t))=\min_{\|y\|=1,y\bot \bold{1}}y'Q(t)y$. Now for every $k=1,\ldots,d$, we can write
\begin{align}\label{eq:estimate-summands-lambda_2(Q)}
\|(I-A(t))\bold{1}^{\!\!\!\!^{\bot (k)}}\|^2&=(\bold{1}^{\!\!\!\!^{\bot (k)}})'(I-A(t))'(I-A(t))\bold{1}^{\!\!\!\!^{\bot (k)}}\cr
&=(\bold{1}^{\!\!\!\!^{\bot (k)}})'Q(t)\bold{1}^{\!\!\!\!^{\bot (k)}}\ge\min_{\substack{\|y\|=1 \\ y\bot \bold{1}}}y'Q(t)y\cr 
&=\lambda_2(Q(t)), 
\end{align}
where the inequality holds, since $\bold{1}'\bold{1}^{\!\!\!\!^{\bot (k)}}=0$ and $\|\bold{1}^{\!\!\!\!^{\bot (k)}}\|=1$. 
Substituting \eqref{eq:estimate-summands-lambda_2(Q)} in \eqref{eq:decrease-sum-quadratics} we get
\begin{align}\label{eq:decrease-lambda_2(Q)} 
\sum_{\ell=1}^{n}\|x_{\ell}(t)-x_{\ell}(t+1)\|^2\ge \sum_{k=1}^{d}\lambda_2(Q(t))\bar{c}^2_k\ge \lambda_2(Q(t))\frac{\epsilon^2}{4}.
\end{align}

Henceforth, we bound $\lambda_2(Q(t))$ from below based on a function of $n$. For that purpose, let us assume that $D(t)=diag\big(1+d_1(t),1+d_2(t),\ldots,1+d_n(t)\big)$, i.e., $D(t)$ is a diagonal matrix with $D_{kk}(t)=1+d_k(t), k\in [n]$. Moreover, let $\mathcal{L}(t)$ denote the Laplacian matrix of the communication graph at time step $t$. By entry wise comparison of both sides, it is not hard to see that $I-A(t)=D(t)^{-1}\mathcal{L}(t)$. Now we can write,
\begin{align}\label{eq:lambda-L2-L}
\lambda_2(Q(t))&=\lambda_2((D(t)^{-1}\mathcal{L}(t))'(D(t)^{-1}\mathcal{L}(t)))\cr 
&=\lambda_2(\mathcal{L}(t)D(t)^{-2}\mathcal{L}(t)),
\end{align} 
where the last equality is due to the fact that $\mathcal{L}(t)$ and $D(t)$ are both symmetric matrices. Next, using the same argument as above, we notice that since $\mathcal{L}(t)D(t)^{-2}\mathcal{L}(t)$ is a symmetric and real-valued matrix, it is diagonalizable, and its zero eigenvalue corresponding to eigenvector $\bold{1}$ has multiplicity one. To see that, let us assume that there is another vector $u$ such that $\mathcal{L}(t)D(t)^{-2}\mathcal{L}(t)u=0$; then, we must have,
\begin{align}\nonumber
0=u'\mathcal{L}(t)D(t)^{-2}\mathcal{L}(t)u=\sum_{i=1}^{n}(\frac{1}{1+d_i(t)})^2(\mathcal{L}(t)u)^2_i,
\end{align} 
which results in $\mathcal{L}(t)u=0$, or, equivalently, $u$ is a scalar multiple of the consensus vector $\bold{1}$. 

\smallskip
Now, using the Courant-Fischer lemma, we can write, 
\begin{align}\label{eq:estimate-L2-I} 
\!\!\!\lambda_2\big(\mathcal{L}(t)D(t)^{-2}\mathcal{L}(t)\big)&\!=\!\min_{\substack{\|y\|=1 \\y\bot \bold{1}}}y'\mathcal{L}(t)D(t)^{-2}\mathcal{L}(t)y\cr 
&\!\ge\! \min_{\substack{\|y\|=1\\ y\bot \bold{1}}}y'\mathcal{L}(t)(\frac{1}{n^2}I)\mathcal{L}(t)y\cr 
&\!=\!\lambda_2\Big(\mathcal{L}(t)(\frac{1}{n^2}I)\mathcal{L}(t)\Big)\cr 
&\!=\!\frac{1}{n^2}\lambda_2\big(\mathcal{L}^2(t)\big)\!=\!\frac{1}{n^2}\lambda_2^2\big(\mathcal{L}(t)\big), 
\end{align}
where the last equality is due to the fact that $\mathcal{L}$ is diagonalizable (it is a symmetric and real-valued matrix) with an eigenvalue 0 of multiplicity 1. 
Substituting \eqref{eq:estimate-L2-I} in \eqref{eq:lambda-L2-L} we get $\lambda_2(Q(t))\ge \frac{1}{n^2}\lambda_2^2\big(\mathcal{L}(t)\big)$. Now, using Cheeger's Inequality (Lemma \ref{lemma:cheegers-inequality}) and since $\mathcal{L}(t)$ is the Laplacian of a connected graph, we can bound $\lambda_2\big(\mathcal{L}(t)\big)$ from below by $\frac{2}{n^2}$, which is due to the isoperimetric number of the communication graph for the minimum cut set. Putting it all together, we have, 
\begin{align}\label{eq:lowerbound-lambda-Q}
\lambda_2(Q(t))\ge \frac{1}{n^2}\lambda_2^2\big(\mathcal{L}(t)\big)\ge \frac{4}{n^6}.  
\end{align}
Finally, combining \eqref{eq:lowerbound-lambda-Q} with \eqref{eq:decrease-lambda_2(Q)}, we conclude that the amount of decrease in the quadratic Lyapunov function if there is a non-$\epsilon$-trivial component is at least $\frac{\epsilon^2}{n^6}$. In other words, if $t$ is not a merging time, we have $V(t)-V(t+1)\ge \frac{\epsilon^2}{n^6}$. Since by definition $V(\cdot)$ is always a nonnegative quantity  with $V(0)\leq \epsilon ^2 n^2$ and the number of merging times can be at most $n$, we conclude that the termination time is bounded from above by $n^8+n$. 
\end{proof}

\bigskip
\section{Asynchronous Hegselmann-Krause Dynamics}\label{sec:mainresults-two}

\bigskip
In this section, we consider the asynchronous Hegselmann-Krause dynamics as introduced in Section \ref{sec:HKdyn}. We first notice that such dynamics do not necessarily reach their steady state in finite time. The simplest case one can consider is when there are only two agents on the real line, separated by a distance less than the confidence bound $\epsilon$. In such a case, no matter what the order of the updating process is, the agents will never arrive at the same opinion or disappear from each other's neighborhood. The two agents will get closer and closer and asymptotically converge to some steady state. That justifies asymptotic analysis of the asynchronous Hegselmann-Krause dynamics, which we will consider in this section. 

\smallskip
In fact, one can easily show that unless the dynamics start from a steady state, it will never reach its steady state in finite time for any asynchronous updating scheme. The reason is that unless the dynamics start from a steady state, at any time instant $t$, there are at least two agents $i$ and $j$ who are connected ($j\in \N_i(t)$), and updating any of them does not bring them to the same opinion. Furthermore, unlike the synchronous case in one dimension, where the order of agents' opinions is preserved throughout the dynamics, in the asynchronous case, the order of the agents' opinions may or may not change, depending on the updating scheme. In this section, we consider a uniformly randomized updating scheme for the agents and analyze the asymptotic convergence of such dynamics to their steady state. But before we start, we need the following two definitions. 

\smallskip
\begin{definition}
We call an updating process a \textit{uniform} updating scheme for the asynchronous Hegselmann-Krause mode if at each time instant $t=0,1,\ldots$, only one agent is chosen independently and with probability $\frac{1}{n}$ from the set of all agents $[n]$ and updates its opinion.   
\end{definition}

\smallskip
\begin{definition}\label{def:delta-equilibrium}
Given a $\delta>0$, we say that an opinion profile $x(t)$ is a $\delta$-equilibrium if the set of agents partition into different sets (clusters) $\{C_1,C_2,\ldots,C_{m}\}$ for some $m\in \mathbb{N}$ such that $dist\big(conv(C_i),conv(C_j))>\epsilon, \forall i\neq j$ and $diam(conv(C_{k}))<\delta, \forall k=1\ldots m$.  
\end{definition}

\smallskip
In fact, Definition \ref{def:delta-equilibrium} simply states that a profile $x(t)$ is a $\delta$-equilibrium if the opinions of agents at time $t$ form some small groups of diameter at most $\delta$ that are far from each other by a distance of at least $\epsilon$. Next, we introduce a network formation game that can explain the behavior of the agents in asynchronous Hegselmann-Krause dynamics.

\bigskip
\subsection{Network Formation Game}

Let us consider a set of $n$ road constructors (players) in $\mathbb{R}^d$ who are funded by the government to construct roads. The budget that the government allocates to each player at the beginning is a fixed amount and is equal to $\$(n-1)\epsilon^2$ ($\$\epsilon^2$ support for each possible road that one player can construct). Ideally, the government would like for all the possible ${n \choose 2}$ roads to be constructed by the players. To that end and in order to create an incentive for players to build as many roads as they can, the government will punish each player by $\$\epsilon^2$ if he or she decides not to construct a road (i.e., the government will take that player's supporting $\$\epsilon^2$ back). On the other hand, each player has the ability to construct roads only within an $\epsilon^2$-neighborhood of himself or herself. (One can assume that the players do not take risks and do not want to spend money beyond the support they received from the government per road.) In such a game, players act myopically, trying to build roads with those who are most beneficial to them. If two players who are located at $x,y\in \mathbb{R}^d$ build a road together, the cost to them is naturally proportional to their distance from each other and is equal to $\|x-y\|^2$. (The farther the players are from each other, the more costly to make a road.) Therefore, in that setting, the payoff for the $i$th player, $i\in [n]$, at location $x_i$ can be formulated as 
\begin{align}\label{eq:payoff}
U_i(x_i,x_{-i})=(n-1)\epsilon^2-\sum_{j=1}^{n}\min\{\|x_i-x_j\|^2,\epsilon^2\},
\end{align}
where $x_{-i}$ denotes the actions of all players except the $i$th one. In such a game, we assume that agents act rationally and are able to compute and play their best response at time steps $t=0,1,2,\ldots$. Based on the above scenario, we have the following lemma. 

\smallskip
\begin{lemma}\label{lemm:best-response}
The sequence of the players' best responses in the network formation game under some specific updating scheme is equivalent to the evolution of the asynchronous Hegselmann-Krause dynamics under the same updating scheme.      
\end{lemma}
\smallskip
\begin{proof}
Let us assume that at time step $t$ the $i$th agent updates his location in order to increase his payoff. If the current locations of the players are denoted by $x_1(t),x_2(t),\ldots, x_n(t)$, the position of agent $i$ at the next time step would be  
\begin{align}\nonumber 
x_i(t+1)&=\argmin_{x} \sum_{j=1}^{n}\min\{\|x-x_j(t)\|^2,\epsilon^2\}\cr 
&=\argmin_{x}\sum_{j\in \N_i(t)}\|x-x_j(t)\|^2=\frac{\sum_{j\in \N_i(t)}x_j(t)}{|\N_i(t)|}.
\end{align}
That establishes the equivalence between the best response dynamics and the updating process in the asynchronous Hegselmann-Krause model.   
\end{proof}

\smallskip
\begin{proposition}
An action profile $(x^*_1,x^*_2,\ldots,x^*_n)$ is a Nash equilibrium of the network formation game if and only if it is a steady state of the asynchronous Hegselmann-Krause dynamics.
\end{proposition}
\begin{proof}
Given an arbitrary Nash equilibrium $(x^*_1,x^*_2,\ldots,x^*_n)$, we show that it is a steady state of the asynchronous Hegselmann-Krause dynamics by showing that for all $i,j\in [n]$ we either have $x^*_i=x^*_j$, or $\|x^*_i-x^*_j\|>\epsilon$. To show this by contradiction, let us assume that there are two players at locations $x^*_p\neq x^*_q$ such that $\|x^*_p-x^*_q\|\leq\epsilon$. Let $L=\{x^*_p,x^*_q,x^*_{\ell_1},\ldots,x^*_{\ell_s}\}$ denote the set of all the players' actions at this equilibrium point which are in the same connected component as $x^*_p$ and $x^*_q$ in the communication graph. Denoting one of the extreme points of $conv(L)$ by $x^*_{\ell}$ and using Lemma \ref{lemm:best-response}, it is not hard to see that player $\ell$'s action is not his best response, i.e., $\frac{\sum_{j\in \N^*_{\ell}}x^*_j}{|\N^*_{\ell}|}\neq x^*_{\ell}$, where $\N^*_{\ell}=\{j:\|x^*_j-x^*_{\ell}\|\leq \epsilon\}$. This is in contrast with the assumption of $(x^*_1,x^*_2,\ldots,x^*_n)$ being a Nash equilibrium. To show that every steady state of the asynchronous Hegselmann-Krause dynamics is a Nash equilibrium of the network formation game is quite straight forward.
\end{proof}

\smallskip
Next we show that the above network formation game is, indeed, a potential game, with the sum of the utilities as a potential function. A further result (Corollary \ref{cor:team-problem}) shows directly that it is strategically equivalent to a team problem.   
\smallskip
\begin{theorem}\label{thm:potential-game}
The network formation game is a potential game with a potential function of $U(x_1,x_2,\ldots,x_n)=\sum_{i=1}^{n}U_i(x_i,x_{-i})$. In particular, we have
\begin{align}\nonumber
U(x_i,x_{-i})-U(x'_i,x_{-i})\leq -2|\N_i|\|x_i-x'_i\|^2,
\end{align}
where $x'_i$ denotes the deviation of the $i$th player from action $x_i$ to his best response $x'_i=\frac{1}{|\N_i|}\sum_{j\in \N_i}x_j$, and $x_{-i}$ denotes the actions of all players except the $i$th one.  
\end{theorem}
\smallskip
\begin{proof}
Let $\N_i$ and $\N'_i$ denote the set of neighbors of player $i$ before and after deviating, respectively. By definition of the payoff function of players \eqref{eq:payoff}, we can write, 
\begin{align}\label{eq:potential-break-sums}
U(x_i,x_{-i})-U(x'_i,&x_{-i})=\!\!\!\!\!\!\sum_{j\in \N_i\cup \N'_i}\!\!\!\!\big(U_j(x_i,x_{-i})-U_j(x'_i,x_{-i})\big)\cr 
&=U_i(x_i,x_{-i})-U_i(x'_i,x_{-i})\cr
&+\!\!\!\!\sum_{j\in \N_i\cap\N'_i}\!\!\!\big(U_j(x_i,x_{-i})-U_j(x'_i,x_{-i})\big)\cr 
&+\!\!\!\!\!\!\sum_{j\in \N_i\setminus\N_i\cap\N'_i}\!\!\!\big(U_j(x_i,x_{-i})-U_j(x'_i,x_{-i})\big)\cr
&+\!\!\!\!\!\!\sum_{j\in \N'_i\setminus\N_i\cap\N'_i}\!\!\!\big(U_j(x_i,x_{-i})-U_j(x'_i,x_{-i})\big),
\end{align}
where the first equality is due to the fact that the utility of the players who do not observe $x_i$ or $x'_i$ does not change. 

Next, we compute each of the summands in the above expression. Note that only the action of player $i$ changes from $x_i$ to $x'_i$, while all others' actions remain unchanged (Figure \ref{fig:FigurePic1}). We can write,
\begin{align}\label{eq:summand-one-potential} 
&U_j(x_i,x_{-i})\!-\!U_j(x'_i,x_{-i})\!=\!\|x_j\!-\!x'_i\|^2\!-\!\|x_j\!-\!x_i\|^2\!\!, j\!\in\! \N_i\!\cap\!\N'_i \cr 
&U_j(x_i,x_{-i})\!-\!U_j(x'_i,x_{-i})\!=\epsilon^2\!-\!\|x_j\!-\!x_i\|^2, j\in\N_i\setminus\N_i\!\cap\!\N'_i \cr 
&U_j(x_i,x_{-i})\!-\!U_j(x'_i,x_{-i})\!=\!\|x_j\!-\!x'_i\|^2\!-\epsilon^2, j\in\N'_i\setminus\N_i\!\cap\!\N'_i. 
\end{align}
\begin{figure}[htb]
\begin{center}
\includegraphics[totalheight=.2\textheight,
width=.3\textwidth,viewport=0 0 450 450]{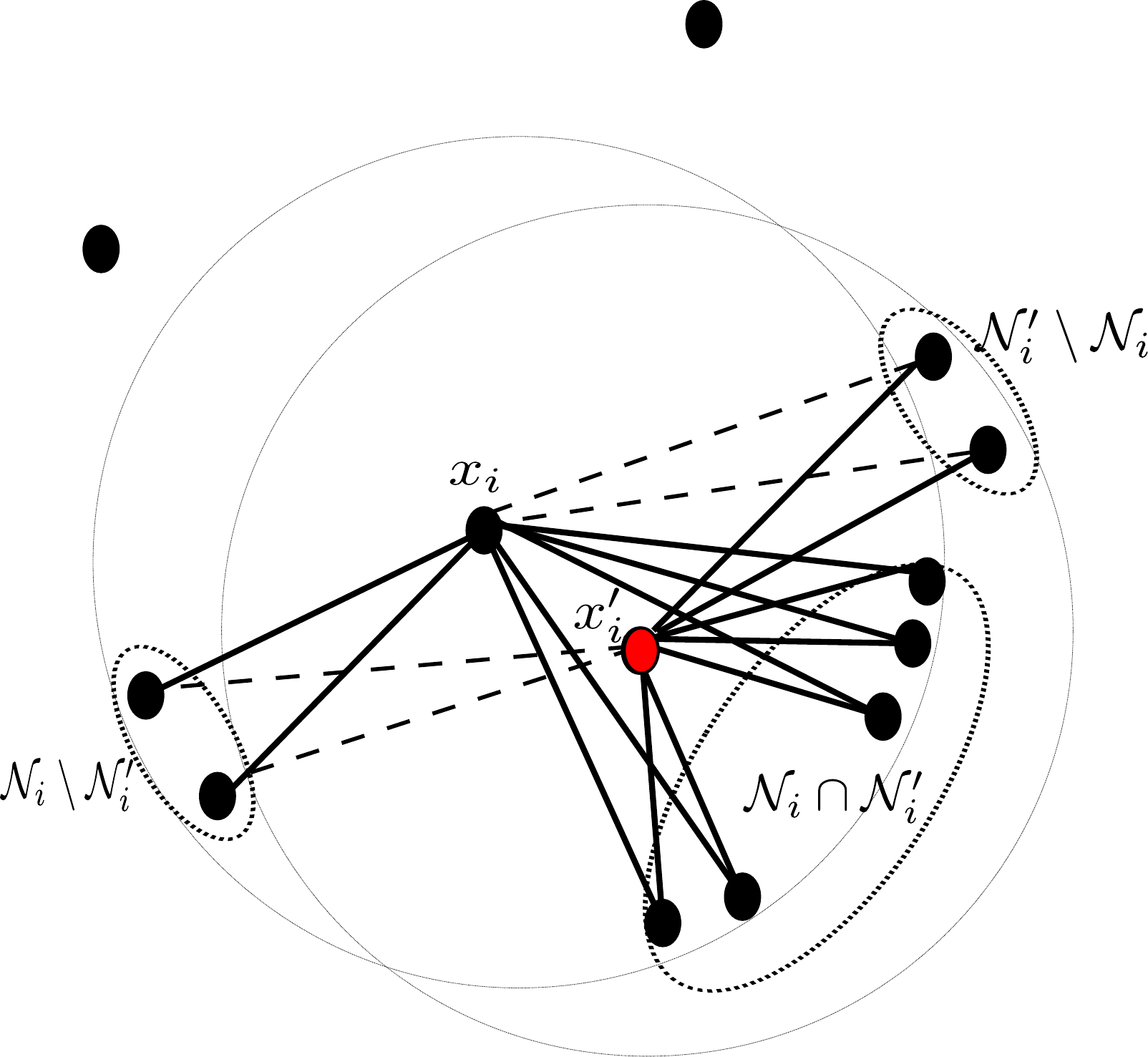} \hspace{0.4in}
\end{center}
\vspace{-0.4cm}
\caption{Deviation of the $i$th player by updating to his best response $x'_i$.}
\label{fig:FigurePic1}
\end{figure}

The reason for the first equality in \eqref{eq:summand-one-potential} is that after the $i$th player deviates, every agent in $j\in\N_i\!\cap\!\N'_i$ still holds his connection with $i$, and hence, by the definition of the payoff function \eqref{eq:payoff}, his payoff is subjected to a change of $\!\|x_j\!-\!x'_i\|^2\!-\!\|x_j\!-\!x_i\|^2$. (Note that all players except the $i$th one are kept fixed.) Similarly, every player $j\in\N_i\setminus\N_i\!\cap\!\N'_i$ stays connected to $x_i$ while disconnecting his link with the $i$th player after $i$'s deviation (since agent $i$ gets far from him by moving from $x_i$ to $x'_i$, and hence they both prefer to stop building the road and each pay $\$\epsilon^2$ to the government). Therefore, the amount of change in the $j$th player's payoff is equal to $\epsilon^2\!-\!\|x_j\!-\!x_i\|^2$. In a similar way, one can observe that the third equality in \eqref{eq:summand-one-potential} holds. By the same line of argument and because of symmetry, one can easily show that the amount of change in the $i$th player's payoff is equal to the sum of all the terms in \eqref{eq:summand-one-potential} over $j\in \N_i\cup \N'_i$. In fact, we can write,
\begin{align}\label{eq:i-utility-change}
&U_i(x_i,x_{-i})\!-\!U_i(x'_i,x_{-i})\!=\!\!\!\!\sum_{j\in \N_i\cap\N'_i}\!\!\!\big(\|x_j\!-\!x'_i\|^2-\|x_j\!-\!x_i\|^2\big)\cr 
&+\!\!\!\!\!\!\!\!\!\sum_{j\in \N_i\setminus\N_i\cap\N'_i}\!\!\!\!\!\!\!\!\!\big(\epsilon^2\!-\!\|x_j-x_i\|^2\big)+\!\!\!\!\!\!\!\!\!\sum_{j\in \N'_i\setminus\N_i\cap\N'_i}\!\!\!\!\!\!\!\!\!\big(\|x_j-x'_i\|^2\!-\!\epsilon^2\big)\cr 
&=(|\N_i|-(|\N_i\cap \N'_i|))\epsilon^2-(|\N'_i|-(|\N_i\cap \N'_i|))\epsilon^2\cr 
&+\sum_{j\in \N'_i}\!\!\!\|x_j\!-\!x'_i\|^2-\sum_{j\in \N_i}\!\!\!\|x_j\!-\!x_i\|^2\cr 
&\leq  \!\!\!\!\!\!\!\!\!\sum_{j\in \N_i\setminus \N_i\cap\N'_i}\!\!\!\!\!\!\!\!\!\!\!\|x_j\!-\!x'_i\|^2\!\!-\!\!\!\!\!\!\!\!\!\!\!\!\!\sum_{j\in \N'_i\setminus \N_i\cap\N'_i}\!\!\!\!\!\!\!\!\!\!\!\|x_j\!-\!x'_i\|^2\!\!+\!\!\!\sum_{j\in \N'_i}\!\!\!\|x_j\!-\!x'_i\|^2\!-\!\!\!\!\sum_{j\in \N_i}\!\!\!\|x_j\!-\!x_i\|^2\cr 
&=\sum_{j\in \N_i}\!\!\|x_j\!-\!x'_i\|^2-\sum_{j\in \N_i}\!\!\|x_j\!-\!x_i\|^2,
\end{align} 
where in the last inequality we have used the facts that 
\begin{align}\nonumber
&(|\N_i|-(|\N_i\cap \N'_i|))\epsilon^2\leq \sum_{j\in \N_i\setminus \N_i\cap \N'_i}\|x_j-x'_i\|^2 \cr  
&(|\N'_i|-(|\N_i\cap \N'_i|))\epsilon^2\ge \sum_{j\in \N'_i\setminus \N_i\cap \N'_i}\|x_j-x'_i\|^2.
\end{align}
(Note that $\|x_j-x'_i\|^2\geq \epsilon^2$, if $j\in \N_i\setminus \N_i\cap \N'_i$, and $\|x_j-x'_i\|^2\leq \epsilon^2$, if $j\in \N'_i\setminus \N_i\cap \N'_i$.)
Substituting \eqref{eq:summand-one-potential} and \eqref{eq:i-utility-change} in \eqref{eq:potential-break-sums} and using \eqref{eq:i-utility-change}, we get 
\begin{align}\nonumber
U(x_i,x_{-i})-U(x'_i,x_{-i})&\leq 2[\sum_{j\in \N_i}\!\!\!\|x_j\!-\!x'_i\|^2-\!\!\sum_{j\in \N_i}\!\!\!\|x_j\!-\!x_i\|^2]\cr 
&=-2|\N_i|\|x_i-x'_i\|^2,
\end{align} 
where the last equality comes from substituting $x'_i=\frac{1}{\N_i}\sum_{j\in \N_i}x_j$ because player $i$ deviates to his best place (Lemma \ref{lemm:best-response}).      
\end{proof}

\smallskip
\begin{corollary}\label{cor:team-problem}
The network formation game is strategically equivalent to a team problem.   
\end{corollary}
\begin{proof}
For any arbitrary player $i\in [n]$, let us define $\beta(x_{-i})\!=\!(n-1)(n-2)\epsilon^2\!-\!\sum_{r,s\in [n]\setminus\{i\}}\min\{\|x_r-x_s\|^2, \epsilon^2\}$. Note that $\beta(x_{-i})$ depends on the actions of all the players except the $i$th player. By definition of $U(x_1,\ldots,x_n)=\sum_{k=1}^{n}U_k(x_k,x_{-k})$, we can write
\begin{align}\nonumber
2U_i(x_i,x_{-i})+\beta(x_{-i})=U(x_1,x_2,\ldots,x_n).
\end{align}
This shows that the network formation game is essentially a team problem, in the sense that every Nash equilibrium of the game is a person-by-person optimal solution for the team, and vice versa. More details on such strategic equivalence can be found in \cite{basar1999dynamic}.         
\end{proof}

Now we are ready to provide an upper bound on the expected number of steps until the asynchronous Hegselmann-Krause dynamics with a uniform updating scheme reaches its $\delta$-equilibrium. 

\smallskip
\begin{theorem}
The expected number of steps until the agents in the asynchronous Hegselmann-Krause dynamics with a uniform updating schedule reach a $\delta$-equilibrium is bounded from above by $2n^9(\frac{\epsilon}{\delta})^2$. 
\end{theorem}
\smallskip
\begin{proof}
We evaluate the expected increase of the potential function given in Theorem \ref{thm:potential-game}. Since each player is chosen independently and with probability $\frac{1}{n}$, we have
\begin{align}\label{eq:expected-decrease-potential}
\mathbb{E}[U(t+1)-U(t)]&=\sum_{i=1}^{n}\frac{1}{n}\mathbb{E}[U(t+1)-U(t)|i\  \mbox{updates}]\cr 
&\ge 2\sum_{i=1}^{n} \frac{|\N_i(t)|}{n}\|x_i(t)-x_i(t+1)\|^2\cr 
&\ge \frac{2}{n}\sum_{i=1}^{n}\|x_i(t)-x_i(t+1)\|^2,
\end{align} 
where in the first inequality we have used the result of Theorem \ref{thm:potential-game}. 

Now using the result of Theorem \ref{thm:termination-time-singletons} and by the same argument as in derivation of \eqref{eq:c-bar-epsilon-trivial}, we know that as long as there is a non-$\delta$-trivial component, we must have $\sum_{k=1}^{d}\bar{c}_k\ge \frac{\delta^2}{4}$, and therefore, $\sum_{i=1}^{n}\|x_i(t)-x_i(t+1)\|^2\ge \frac{\delta^2}{n^6}$. Moreover, since $U(\tau)< n^2\epsilon^2$, we conclude that the expected number of times that nontrivial components of a diameter larger than $\delta>0$ will emerge is bounded from above by $2n^9(\frac{\epsilon}{\delta})^2$.      
\end{proof}

\smallskip
In fact, in the case of scalar asynchronous Hegselmann-Krause dynamics, one could come up with a sharper bound which we state in the following lemma. 
\smallskip 
\begin{lemma}\label{lemm:scalar-asynchronous}
The expected number of steps until the scalar asynchronous Hegselmann-Krause dynamics reach an $\frac{\epsilon}{n}$-equilibrium is bounded from above by $n^{5+2\log_{n}(n+1)}+n$.
\end{lemma}
\smallskip
\begin{proof}
Consider a particular time instant $t$, and let $x_1(t)=\min_{k\in [n]}{x_k(t)}$ and $x_m(t)=\max_{k\in \N_1(t)}x_k(t)$. Also, without loss of generality, let us assume that $x_1(t)=0$. It is clear that if $x_m(t)>\frac{\epsilon}{n^{\alpha}}$ and agent 1 updates, then we will have $x_1(t+1)>\frac{\epsilon}{n^{1+\alpha}}$, where $\alpha$ is a number to be determined later. In this case, the expected potential function will increase by at least $\frac{2}{n}\|x_1(t)-x_1(t+1)\|^2\ge \frac{2\epsilon^2}{n^{3+2\alpha}}$. Otherwise, there is no other agent in the interval $[\frac{\epsilon}{n^{\alpha}}, \epsilon]$. Now we consider two cases (Figure \ref{fig:scalar-asynchronous}):
\begin{itemize}
\item Agent $x_m(t)$ has a neighbor in the interval $(\epsilon, x_m(t)+\epsilon]$. Assuming that agent $m$ updates, we will have $x_m(t+1)\ge \frac{x_m(t)+\epsilon}{n}$, and hence, 
\begin{align}\nonumber
\|x_m(t+1)-x_m(t)\|^2\ge \|\frac{\epsilon}{n}-x_m(t)\|^2\ge(\frac{\epsilon}{n}-\frac{\epsilon}{n^{\alpha}})^2.  
\end{align} 
Therefore, in this case and using \eqref{eq:expected-decrease-potential}, the amount of increase in the expected potential function is at least $\frac{2\epsilon^2}{n^3}(1-\frac{1}{n^{\alpha-1}})^2$. 
\item Agent $x_m(t)$ does not have any neighbor in the interval $(\epsilon, x_m(t)+\epsilon]$. We note that all the agents in the interval $[0, x_m(t)]$ form a cluster that is separated from other agents by a distance of at least $\epsilon$. Noting that two separate clusters of nodes on a real line will stay apart from each other in the rest of the dynamics, we can decompose the original dynamics into two groups and analyze each of them separately. 
\end{itemize}

\begin{figure}[htb]
\vspace{-2cm}
\begin{center}
\hspace{-0.5cm}\includegraphics[totalheight=.25\textheight,
width=.35\textwidth,viewport=0 0 450 450]{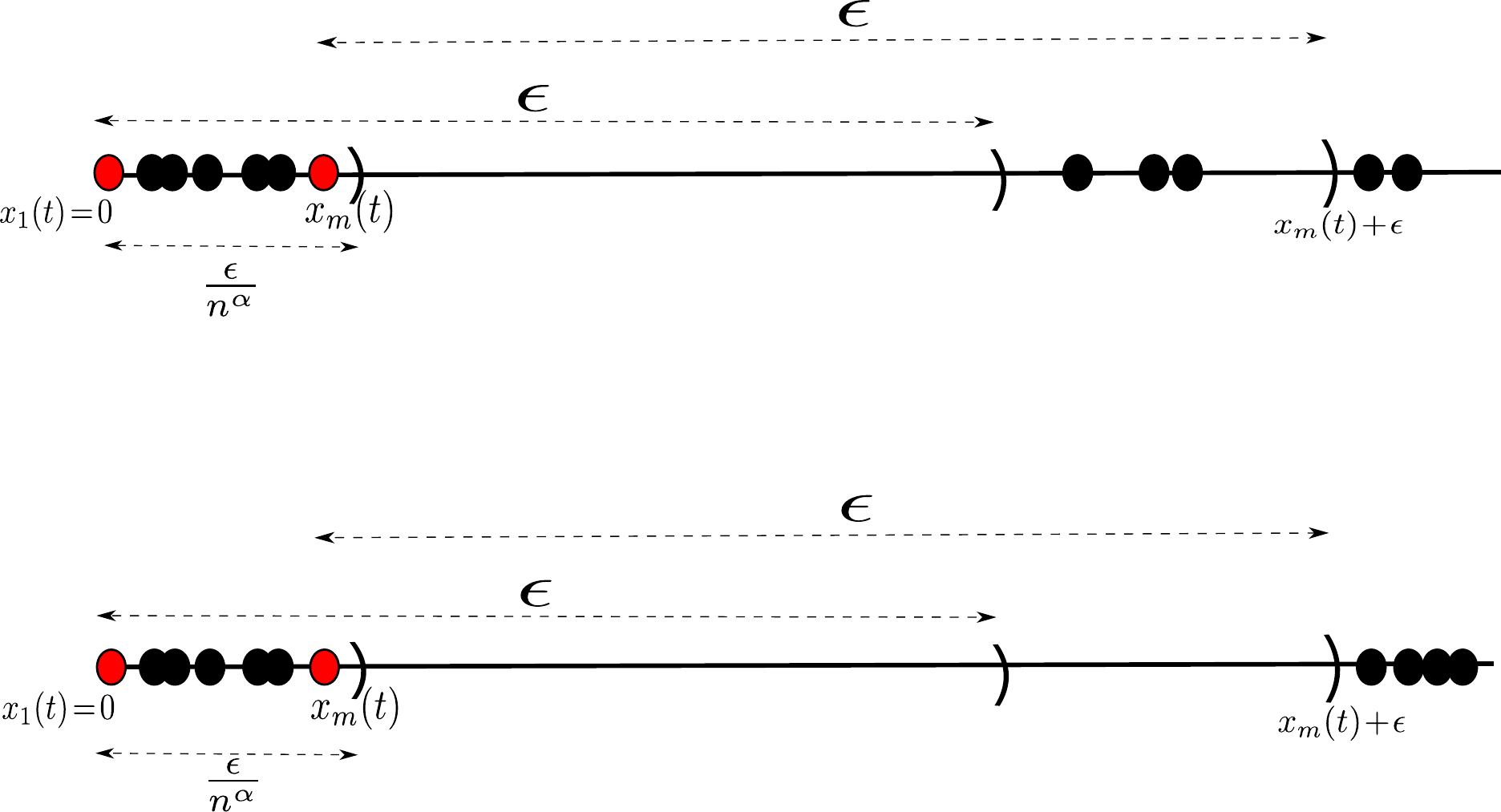} \hspace{0.4in}
\end{center}
\vspace{-0.4cm}
\caption{Illustration of two different cases in the proof of Lemma \ref{lemm:scalar-asynchronous}.}
\label{fig:scalar-asynchronous}
\end{figure}

By choosing $\alpha=\log_{n}(n+1)$, we get $\frac{2\epsilon^2}{n^{3+2\alpha}}=\frac{2\epsilon^2}{n^3}(1-\frac{1}{n^{\alpha-1}})^2$, and we can see that either we have an increase of size $\frac{2\epsilon^2}{n^{3+2\log_{n}(n+1)}}$ in the expected potential function, or the dynamics decompose into a cluster of size at most $\frac{\epsilon}{n^\alpha}<\frac{\epsilon}{n}$ and another part. Since the expected potential function cannot increase more than $n^{5+2\log_{n}(n+1)}$ number of steps ($U(\cdot)\leq n^2\epsilon^2$) and we cannot have more than $n$ clustering decompositions, the expected number of steps until the dynamics decompose to clusters whose size is at most $\frac{\epsilon}{n}$ is bounded from above by $n^{5+2\log_{n}(n+1)}+n$.      
\end{proof}

\smallskip
\begin{remark}
From the above lemma, after the expected number of $n^{5+2\log_{n}(n+1)}+n\approx n^7$, every agent lies within a cluster of diameter at most $\frac{\epsilon}{n}$, and those always are separated from each other by a distance of at least $\epsilon$. Therefore, each agent in a cluster can observe the others, and henceforth, the diameter of the convex hull of each of the clusters shrinks very fast. 
\end{remark}

In the following, we provide a bound on the expected number of switching topologies during the evolution of the asynchronous Hegselmann-Krause process. 

\bigskip
\begin{theorem}\label{thm:expected-switching}
The expected number of switching topologies of the asynchronous Hegselmann-Krause dynamics with a uniform updating scheme is bounded from above by $16n^9$.
\end{theorem}
\smallskip
\begin{proof}
We show that switching topologies substantially increase the expected value of the potential function. To see that, first assume that the opinion profile at time $t-1$, i.e., $x(t-1)$, is $\frac{\epsilon}{2}$-trivial, and that updating some agent $i$ at this time causes a switch in the topology of the network. We claim that the next profile, i.e., $x(t)$, is not $\frac{\epsilon}{2}$-trivial. Note that since there is a switch at time $t$ and that within each of the $\frac{\epsilon}{2}$-trivial components each agent is able to observe the others, the convex hull of such a component shrinks even further after the updating of any agent in the component. Therefore, the switches must occur between the $\frac{\epsilon}{2}$-trivial components and not within them.

Now, let us assume that $i\in C_p$ ($C_p$ denotes an $\frac{\epsilon}{2}$-trivial component) and that updating agent $i$ at time $t-1$ makes him visible to another agent $j$ in a different $\frac{\epsilon}{2}$-trivial component $C_q$ (Figure \ref{fig:FigurePic2}). Since $C_p$ is an $\frac{\epsilon}{2}$-trivial component and the agents in $C_p$ are all the agents who are visible to agent $i$ at time $t-1$, the movement of agent $i$ from $x_i(t-1)$ to $x_i(t)$ can be at most $\frac{\epsilon}{2}$. Moreover, since agents $j$ and $i$ belong to different $\frac{\epsilon}{2}$-trivial components, their distance at time $t-1$ was larger than $\epsilon$. That means that such a switching causes $i$ and $j$ to make a link with a distance of at least $\frac{\epsilon}{2}$ in the profile $x(t)$.

\begin{figure}[htb]
\begin{center}
\includegraphics[totalheight=.2\textheight,
width=.3\textwidth,viewport=0 0 500 500]{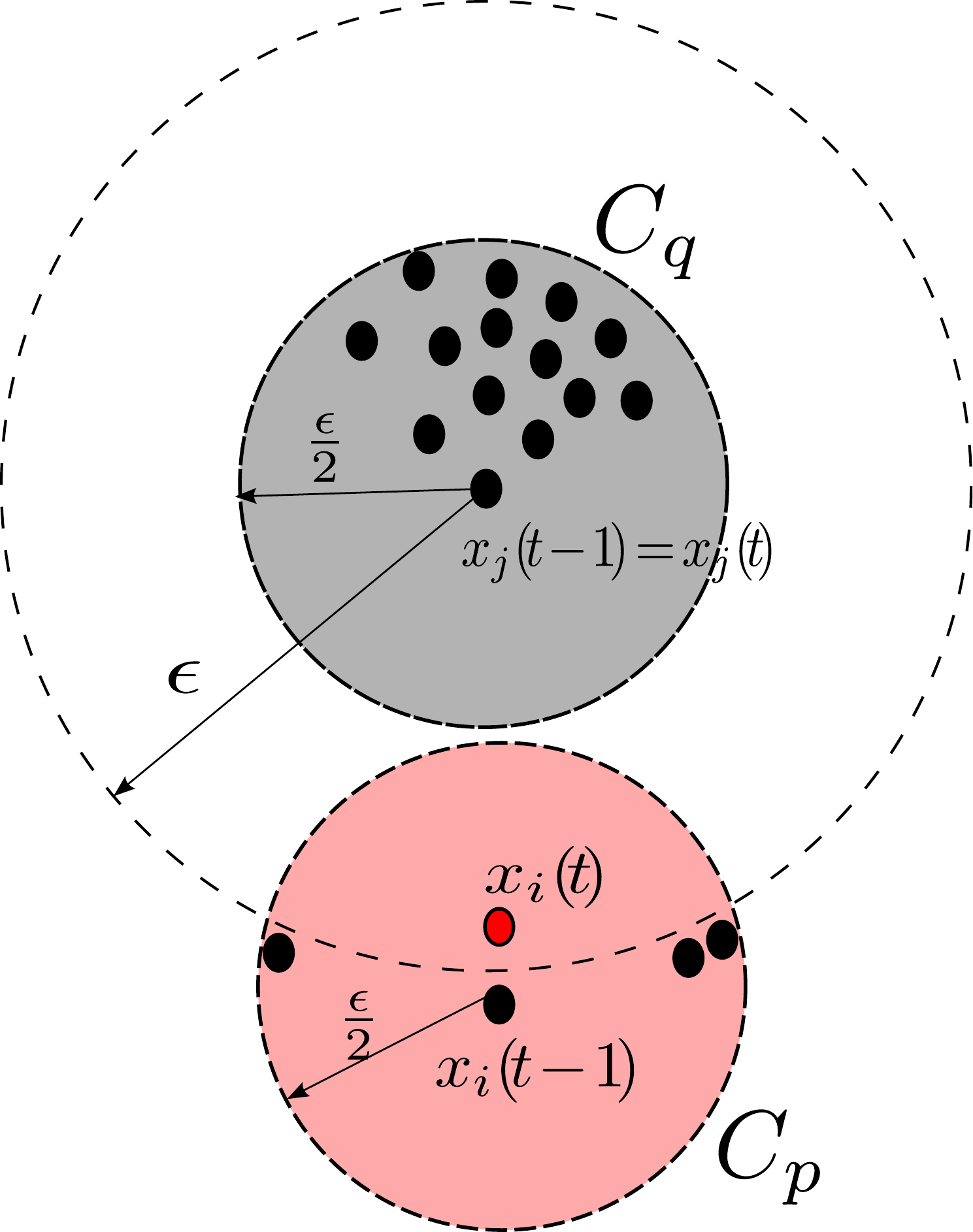} \hspace{0.4in}
\end{center}
\vspace{-0.4cm}
\caption{Switching topology at time $t$ from an $\frac{\epsilon}{2}$-trivial profile $x(t-1)$.}
\label{fig:FigurePic2}
\end{figure}

Now we partition all the possible switching times based on the profile at the previous time instant: 
\smallskip
\begin{itemize}
\item Time $t$ is a switching time, and $x(t-1)$ is an $\frac{\epsilon}{2}$-trivial profile. In this case and based on the above argument, $x(t)$ is not an $\frac{\epsilon}{2}$-trivial profile, and using the same argument as in relation \eqref{eq:c-bar-epsilon-trivial} and in view of \eqref{eq:decrease-lambda_2(Q)} and \eqref{eq:lowerbound-lambda-Q}, we get $\sum_{k=1}^{n}\|x_k(t)-x_k(t+1)\|^2\ge \frac{\epsilon^2}{16n^6}$. 

\smallskip 
\item Time $t$ is a switching time, and $x(t-1)$ is not an $\frac{\epsilon}{2}$-trivial profile. In this case and within a non-$\frac{\epsilon}{2}$-trivial component, using the same argument as in the first case, we get $\sum_{k=1}^{n}\|x_k(t-1)-x_k(t)\|^2\ge \frac{\epsilon^2}{16n^6}$.        
\end{itemize}

\smallskip
Therefore, if $t$ is a switching time, using \eqref{eq:expected-decrease-potential} we conclude that there is an increase of $\frac{\epsilon^2}{16n^6}$ at either time $t-1$ or $t$ in the expected potential function. In other words, if $t$ is a switching time, using \eqref{eq:expected-decrease-potential} we can write, 
\begin{align}\nonumber
\mathbb{E}[U(t\!+\!1)\!-\!U(t\!-\!1)]&=\mathbb{E}[U(t\!+\!1)\!-\!U(t)]\cr 
&+\mathbb{E}[U(t)\!-\!U(t\!-\!1)]\ge \frac{2}{n}\frac{\epsilon^2}{16n^6}=\frac{\epsilon^2}{8n^7}.
\end{align}

Now, given an arbitrary initial profile $x(0)$, let us use $p_t$ to denote the probability of occurrence of a switching at time $t=1,2,\ldots$. Therefore, the amount of increase in the expected potential function is at least $\sum_{t=0}^{\infty}p_t\frac{\epsilon^2}{16n^7}$ (since we may count each instant twice). On the other hand, since $U(\tau)\leq n^2\epsilon^2, \forall \tau=1,2,\ldots$, we conclude that $\sum_{t=0}^{\infty}p_t$. But $\sum_{t=0}^{\infty}p_t$ is exactly equal to the expected number of switching topologies. Therefore, the expected number of switching topologies is bounded from above by $16n^9$. 
\end{proof}
   
\bigskip                
\section{Heterogeneous Hegselmann-Krause Dynamics}\label{sec:mainresults-three}

\smallskip
Once again we consider the Hegselmann-Krause model \eqref{eq:synchronous-Hegselmann-Krause}, but this time we assume that each agent $i$ has his or her own bound of confidence $\epsilon_i$, which could be different from the others. Therefore, $\N_i(x(t))=\{1\leq j \leq n: \|x_i(t)-x_j(t)\|\leq \epsilon_i \}$ and $A(t), t\ge 0$ will change correspondingly. That causes an asymmetry for the interactions among the agents. In other words, there is a possibility that one agent $x_i(t)$ observes agent $x_j(t)$ but not vice versa. In fact, we are interested in studying the convergence behavior of such dynamics. In contrast with the homogeneous Hegselmann-Krause model, which reaches its steady state after finite time, the following example shows that in the heterogeneous case, steady state may not be reached in finite time. 
\smallskip
\begin{example}
Consider three agents $x_1, x_2, x_3$ that are located at $-1, \frac{1}{3}, 1$,
respectively, at the initial time $t = 0$. Also, let us assume $\epsilon_1 = \frac{1}{2}, \epsilon_2 = 2, \epsilon_3 =\frac{1}{2}$. As can be seen, agent $x_2$ is able to see all the agents at each time
step. Therefore, after the first iteration, $x_2(1)=\frac{-1+\frac{1}{3}+1}{3}=\frac{1}{3^2}$, and since the confidence bounds of $x_1$ and  $x_3$ are small, they can see no one except themselves, and hence they will remain in their own locations. Therefore, at time $t = 1$, we will have $x_1(1)=-1, x_2(1)=\frac{1}{3^2}, x_3(1) = 1$. With the same line of argument, it is not hard to see that at any time instant $t=1,2,\ldots$ the position of agents will be $x_1(t)=-1, x_2(t)=\frac{1}{3^{t+1}}, x_3(t)=1$. That shows that the dynamics will converge to their steady state $(-1, 0, 1)$, but not in finite time.
\end{example}

\smallskip
In the above example, one of the main reasons that the convergence was
not achieved in finite time was that there were two agents who didn't have interaction with others in the dynamics and remained fixed without any movement forever. We refer to such agents as \textit{silent agents}. In the next theorem, we show that if the amount of time an agent sleeps (is inactive) is finite, then we will have finite time convergence of the dynamics to their steady state. We note that similar type of such asynchronous analysis under different scenarios and settings can be found in \cite{li1987asymptotic,bertsekas1983distributed,tsitsiklis1984convergence}. 

\bigskip
\begin{theorem}
Consider the heterogeneous Hegselmann-Krause model, where the $i$th agent $i\in [n]$ has a confidence bound of $\epsilon_i > 0$. Also, assume that there is an integer $T^*$ such that no agent is silent for a period of time longer than $T^*$. Then, the dynamics will converge to their steady state in finite time.
\end{theorem}
\smallskip
\begin{proof}
We prove the theorem by induction on the number of agents. For $n=1$
the result is obvious, and the initial time is the termination time. Let us assume that the result holds for each $k\leq n$, and now suppose that we have $n+1$ agents with different confidence bounds. We show that there is a finite time $T$ such that the left product of every $T$ consecutive matrices $A(t), t\ge 0$ of the dynamics will generate a matrix with at least one positive column.  

\smallskip
Starting from agent 1, let us define 
\begin{align}\nonumber
S(t)=\{i\in [n+1]|(A(t)A(t-1)\ldots A(0))_{i1}>0\},
\end{align}
and $S^{c}(t)=[n+1]\setminus S(t)$ to be its complement. Since each agent can see itself at each time instant, if $i\in S(t)$ for some time $t$, then it will be in $S(t')$ for all $t'\ge t$. In other words, we have $S(0)\subseteq S(1)\subseteq S(2)\subseteq \ldots$. Now we claim that there must be a finite time $T$ such that $S(T)=[n+1]$. Otherwise, let us assume that there exists a time instant $t_0$ such that $S(t_0)=S(t), \forall t>t_0$. By the definition of $S(t)$, that means that for $t>t_0$, none of the agents in $S^c(t)$ can see any agent in $S(t)$ (although it may happen that some agents in $S(t)$ are still able to see some of the agents in $S^c(t)$). That means that the agents in the set $S^c(t_0)$ constitute a group of agents whose opinions in the future of the dynamics $t\ge t_0$ will not be influenced by any other agent in $S(t_0)$. On the other hand, since $|S^c(t_0)|\leq n$ (note that $S(0)=\{1\}$), according to the induction assumption, the agents in $S^c(t_0)$ will reach their steady state after some finite time $T_n$, where $T_n$ denotes the maximum number of steps for $n$ agents to reach their steady state, which, by induction assumption, is considered to be a finite number. However, under the hypothesis of the Theorem, after reaching the steady state, these agents cannot remain silent for more than $T^*$ more steps. Therefore, after a finite time $T^* +T_n$, at least one more agent will be added to the set $S(t_0)$, and the cardinality of $S(t_0)$ will increase by at least 1. Since the total number of agents is $n+1$, $T:=(n+1)(T^*+T_n)$ steps are enough to guarantee $S(T)=[n+1]$. That shows that $A(T)\ldots A(1)A(0)$ will be a matrix in which the first column will be strictly positive. 

On the other hand, since all the positive entries of those matrices are bounded from below by $\min^+(A(t))\ge \frac{1}{n+1}$, the minimum positive entry of the left product of every $T$ consecutive such matrices will be larger than $(\frac{1}{n+1})^{T}$. Using Lemma \ref{lemm:convex-hull}, we can see that after every $T$ steps, the diameter of the convex hull of the agents' opinions will shrink by a factor of at least $1-(\frac{1}{n+1})^{T}$. Therefore, there exists a finite time $T_{n+1}<\infty$ such that the diameter of the convex hull of the agents' opinions at time $T_{n+1}$ is smaller than $\min_{i\in [n+1]}\epsilon_i$. That means that after $T_{n+1}$ steps, every agent is able to observe the others in his or her own neighborhood, and in the next step, the dynamics reach a steady state.    
\end{proof}

\smallskip
In fact, the above theorem asserts that if there exists an external input which creates an incentive for the agents to interact with someone else after some period of time, then the circulation of information in the society will be sufficient to guarantee the finite time formation of the opinions.    

\bigskip
\section{Discussion}\label{sec:discussion}
Inspired by the results given in Section \ref{sec:mainresults-two}, we will now discuss some of the possible directions that could be pursued to analyze the asynchronous heterogeneous Hegselmann-Krause model in more detail. In fact, because of the different confidence bounds, the symmetry from which we benefit in the homogeneous case does not hold anymore. Therefore, the communication topology in this case can be interpreted as a digraph (directed graph) instead of an undirected graph. In this case one way of showing the asymptotic convergence of the heterogeneous Hegselmann-Krause dynamics to an steady state is to design a proper utility function for each player such that the resulting network formation game changes to a team problem, such that each player's update contributes an increase (decrease) to a global function toward an equilibrium.

\smallskip
A natural idea here is to define the utility of the players based on functions of their own confidence bound and their relative distance from others such that their best response dynamics coincide with the evolution of the asynchronous heterogeneous Hegselmann-Krause dynamics. For example, one may define the utility of the $i$th player to be $U_i(t)=(n-1)\epsilon_i^2-\sum_{j=1}^{n}\min\{(x_i(t)-x_j(t))^2,\epsilon^2_i\}$, where $\epsilon_i$ denotes the confidence bound of the $i$th agent and $x(t)=(x_1(t),x_2(t),\ldots,x_n(t))$ denotes the opinion profile at time instant $t$. It turns out that such utility functions do not make the network formation game a potential game or lead to a strategically equivalent team problem. However, one can consider 15 different possibilities for creation or breaking of edges among agents, assuming that only one agent updates (deviates) to a new position. In that case, one can think of a proper weighting on the edges in order to \textit{distinguish} one-sided edges from symmetric (two-sided) edges. For example, if there is a one-sided edge from player $i$ to player $j$, one can rescale the utility of agent $i$ by a fraction of his own confidence bound and his neighbors' in order to adjust the influence of other players' actions on his own utility function. At this point, we are not aware of any such utility functions, and we leave the full analysis of the heterogeneous Hegselmann-Krause dynamics as a future direction of research.          
   
\bigskip
\section{Conclusion}\label{sec:conclusion}
In this paper, we studied the termination time of the Hegselmann-Krause dynamics in finite dimensions and under various settings: synchronous, asynchronous, homogeneous, and heterogeneous. We provided a polynomial upper bound for the termination time of the synchronous homogeneous model independent of the dimension of the ambient space. We showed that the asynchronous Hegselmann-Krause model can be formulated as a sequence of best response dynamics of a potential game. Furthermore, we provided an upper bound for the expected number of steps until the dynamics reaches its $\delta$-equilibrium. In particular, we bounded the expected number of switchings in the topology of the networks during the evolution of the system. We considered the heterogeneous Hegselmann-Krause dynamics, and we obtained a necessary condition for finite time convergence of such dynamics. Finally, we discussed some of the possible future directions that could be pursued to enable analysis of heterogeneous Hegselmann-Krause dynamics in more detail. As a future direction of research, one may think of how to enrich the Hegselmann-Krause model in order to remove some of its current limitations. As an example, one could modify the model by allowing the agents with the same opinion to be related to each other by some constraints, meaning that having the same opinion at some time instant does not necessarily lead to having the same opinion for all the future time instances.     

\bibliographystyle{IEEEtran}
\bibliography{thesisrefs}
\end{document}